\documentclass[twoside,11pt]{article}
\usepackage{jair,theapa}
\usepackage{times}
\usepackage{helvet}
\usepackage{courier}

\usepackage{graphicx}

\usepackage{rawfonts}
\usepackage{url}

\usepackage{bbm}
\usepackage{dsfont}
\usepackage{amsfonts}
\usepackage{amsmath,amssymb,amsthm}
\usepackage{booktabs}
\usepackage{float}

\newcommand{\bbbr}{\ensuremath \mathbb{R}}
\newcommand{\Neigh}{\ensuremath \mathcal{N}}
\newcommand{\citep}[1]{\cite{#1}}
\newcommand{\citet}[1]{\citeA{#1}}
\newcommand{\citetp}[1]{\citeR{#1}}

\newcommand{\polylog}{\ensuremath \mathrm{polylog}}
\newcommand{\pa}{\ensuremath \mathrm{pa}}
\newcommand{\Ch}{\ensuremath \mathrm{Ch}}
\newcommand{\Mtilde}{\ensuremath \widetilde{M}}
\newcommand{\Itilde}{\ensuremath \widetilde{I}}
\newcommand{\C}{\ensuremath \mathcal{C}}

\newcommand{\Proj}[1]{\ensuremath \mathrm{Proj}\left( #1 \right)}
\newcommand{\Proji}[1]{\ensuremath \mathrm{Proj}_i\left( #1 \right)}

\newcommand{\argmax}{\ensuremath \arg\max}
\newcommand{\argmin}{\ensuremath \arg\min}

\newtheorem{definition}{Definition}
\newtheorem{theorem}{Theorem}
\newtheorem{lemma}{Lemma}
\newtheorem{claim}{Claim}
\newtheorem{corollary}{Corollary}

\ShortHeadings{FPTAS for MSNE in 
  Tree GGs}
{Ortiz \& Irfan}
\firstpageno{1}

\begin{document}

\title{FPTAS for Mixed-Strategy Nash Equilibria in Tree Graphical Games and Their Generalizations}
\author{\name Luis E. Ortiz \email leortiz@umich.edu \\
       \addr Department of Computer and Information Science\\
        University of Michigan - Dearborn\\
        Dearborn, MI 48128 
\AND
\name Mohammad T. Irfan \email mirfan@bowdoin.edu \\
\addr Department of Computer Science\\
        Bowdoin College\\
        Brunswick, ME 04011}

\maketitle
\begin{abstract}
We provide the first fully polynomial time approximation scheme (FPTAS) for computing an approximate mixed-strategy Nash equilibrium in tree-structured graphical multi-hypermatrix games (GMhGs). GMhGs are generalizations of normal-form games, graphical games, graphical polymatrix games, and hypergraphical games. Computing an exact mixed-strategy Nash equilibria in graphical polymatrix games is PPAD-complete and thus generally believed to be intractable. In contrast, to the best of our knowledge, we are the first to establish an FPTAS for tree polymatrix games as well as tree graphical games when the number of actions is bounded by a constant. As a corollary, we give a quasi-polynomial time approximation scheme (quasi-PTAS) when the number of actions is bounded by the logarithm of the number of players.
\end{abstract}

\section{Introduction}

For over a decade, graphical games have been at the forefront of computational game theory. In a graphical game, a player's payoff is directly affected by her own action and those of her neighbors. This large class of games has played a critical role in establishing the hardness of computing a Nash equilibrium in general games~\cite{DaskalakisCACM}. It has also generated a great deal of interest in the AI community since~\citeA{Kearns_et_al_2001} drew a parallel with probabilistic graphical models in terms of succinct representation by exploiting the network structure. As a result, this is one of the select topics in computer science that has triggered a confluence of ideas from the theoretical computer science and AI communities. 

This paper contributes to this development by providing the first
fully polynomial-time approximation scheme (FPTAS) for approximate
Nash equilibrium computation in a generalized class of tree graphical games.
Tree-structured interactions
are natural in hierarchical settings. As often visualized in the
ubiquitous 
organizational chart of bureaucratic structures~\cite{weber1946}, hierarchical
organizations are arguably the most common managerial structures still
found
around the world, particularly in large corporations and governmental institutions (e.g.,
military), as well as in many social and religious
institutions. Supply chains are also commonplace, such as in agriculture~\cite<see,
e.g.,>{doi:10.1080/13675567.2010.518564}. Even within the context of
energy grids, the traditional electric power generation,
transmission, and distribution systems are tree-structured, and are
commonly 
modeled mathematically and computationally as such~\cite<see>[for a recent example]{Dvijotham2016}. 

\emph{Our algorithm eliminates the exponential dependency on the maximum degree of a node, a problem that has plagued research for 15 years since the inception of graphical games~\citep{Kearns_et_al_2001}.}

More generally, we consider the problem of \emph{computing approximate
  MSNE in GMhGs}, as defined by~\citeA{ortiz14}. We refer the reader
to Table~\ref{tab:acro} for a list of acronyms used throughout this paper. Roughly
speaking, in a GMhG, each player's payoff is the summation of several
local payoff hypermatrices defined with respect to each individual player's local
hypergraph. GMhGs generalize normal-form games, graphical games~\citep{Kearns_et_al_2001,Kearns_2007}, graphical polymatrix games, and
hypergraphical games~\citep{Papadimitriou:2008:CCE:1379759.1379762}. For approximate MSNE, we adopt the standard notion of $\epsilon$-MSNE (also known as $\epsilon$-approximate MSNE), an additive (as opposed to relative) approximation scheme widely used in algorithmic game theory~\citep{Lipton_et_al_2003,Daskalakis_et_al_2007,deligkas2014,barman15}.

In this paper, we provide FPTAS and quasi-PTAS for GMhGs in which the individual player's
number of actions $m$ and the hypertree-width $w$ of the underlying game hypergraph
are bounded. The key to our solution is the formulation of a CSP such that any solution to this CSP is an $\epsilon$-MSNE of the game. This raises two challenging questions: Will the CSP have any solution at all? In case it has a solution, how can we compute it \emph{efficiently}? Regarding the first question, we discretize \emph{both} the probability space and the payoff space of the game to guarantee that for \emph{any} MSNE of the game (which always exists), the nearest grid point is a solution to the CSP. For the second question, we give a DP algorithm that is an FPTAS when $m$ and $w$ are bounded by a constant. 
Most remarkably, this algorithm eliminates the exponential dependency on the largest neighborhood size of a node, which has plagued previous research on this problem.

\section{Related Work}

In this section, we provide a brief overview of the previous computational
complexity and algorithmic results for the problem of $\epsilon$-MSNE computation (additive approximation scheme as most commonly defined in game theory) in
general.
A full account of all specific sub-classes of GMhGs such as normal-form games and
(standard) graphical games is beyond the scope of this paper, just as is
the discussion on (a) other types of approximations such as the less common relative approximation; (b) other popular equilibrium-solution concepts such as pure-strategy Nash equilibria and correlated equilibria~\citep{aumann74,aumann87}; and (c) other quality guarantees of solutions, 
including exact MSNE and ``well-supported'' approximate MSNE. 

\begin{table}
\centering
\small
\begin{tabular}{ll}
\toprule
CSP & Constraint Satisfaction Problem \\ 
DP  & Dynamic Programming   \\  
FPTAS & Fully Polynomial Time Approx. Scheme \\ 
GMhG & Graphical Multi-hypermatrix Game \\ 
MSNE & Mixed-Strategy Nash Equilibrium \\ 
Quasi-PTAS & Quasi-Polynomial Time Approx. Scheme \\ 

\bottomrule
\end{tabular}
\caption{Acronyms used in this paper.}
\label{tab:acro}
\end{table}
The complexity status of normal-form games is well-understood
today, thanks to 
a series of seminal works~\citep{DaskalakisCACM,daskalakis:195} that culminated in the
PPAD-completeness of 2-player multi-action normal-form games, also known as
\emph{bimatrix games}~\citep{ChenSettlingJACM}.
Once the complexity of exact MSNE computation
was established, the spotlight naturally fell on approximate MSNE, 
 especially in succinctly representable games such as
graphical games. \citet{ChenSettlingJACM} showed that bimatrix games do not admit an
FPTAS unless PPAD $\subseteq$ P. This result opened up 
computing a PTAS. 

There has been a series of results based on 
constant-factor approximations. The current best PTAS is a 0.3393-approximation for bimatrix
games~\citep{tsaknakis2008}, which can be extended to the cases of
three and four-player games with the approximation guarantees of
0.6022 and 0.7153, respectively. Note that sub-exponential algorithms
for computing $\epsilon$-MSNE in games with a constant number of players
have been known prior to all of these
results~\citep{Lipton_et_al_2003}. As a result, it is unlikely that
the case of constant number of players will be
PPAD-complete.  
Along that line, \citet{rubinstein15} considered the hardness of computing
$\epsilon$-MSNE in $n$-player succinctly representable games such as
general graphical games and graphical polymatrix games. He showed that there exists a constant $\epsilon$ such that finding an $\epsilon$-MSNE in a $2$-action graphical polymatrix game with a bipartite structure and having a maximum degree of 3 is PPAD-complete. 
\citet{ChenSettlingJACM} showed the hardness of bimatrix games for a polynomially small $\epsilon$, and  \citet{rubinstein15} showed the hardness (in this case, PPAD-completeness) of $n$-player polymatrix games for a constant $\epsilon$. 

On a positive note, 
\citet{deligkas2014} presented an algorithm for computing a $(0.5+\delta)$-MSNE of an $n$-player polymatrix game. Their algorithm runs in time polynomial in the input size and $\frac{1}{\delta}$. 
Very recently, \citet{barman15} gave a quasi-polynomial time randomized algorithm for computing an $\epsilon$-MSNE in tree-structured polymatrix games. They assumed that the payoffs are normalized so that the local payoff of any player $i$ from any other player $j$ lies in $[0, 1/{d_i}]$, where $d_i$ is the degree of $i$. This guarantees, in a strong way, that the total payoff of any player is in $[0,1]$. In comparison, we do not make the assumption of local payoffs lying in $[0, 1/{d_i}]$. Also, our algorithm is a deterministic FPTAS when $m$
is bounded by a constant.

Closely related to our work, \citet{ortiz14} gave a framework for sparsely discretizing probability spaces in order to compute $\epsilon$-MSNE in tree-structured GMhGs. The time complexity of the resulting algorithm depends on $(\frac{k}{\epsilon})^k$ when $m$ is bounded by a constant. Ortiz's result is a significant step forward compared to \citet{Kearns_et_al_2001}'s algorithm in the foundational paper on graphical games. In the latter work, the time complexity depends on $(\frac{2^k}{\epsilon})^k$ when $m$ is bounded by a constant. Both of these algorithms are exponential in the representation size of succinctly representable games such as graphical polymatrix games. Compared to these works, our algorithm eliminates the exponential dependency on $k$. Furthermore, compared to Ortiz's work, we discretize \emph{both} probability and payoff spaces in order to achieve an FPTAS. This joint discretization technique is novel for this large class of games and has a great potential for other types of games.

\subsubsection*{Hardness of Relaxing Key Restrictions.}
We use two restrictions:
(1) Our focus is on GMhGs (e.g., graphical polymatrix games) with tree structure, and (2) our FPTAS for $\epsilon$-MSNE computation hinges on the assumption that the number of actions is bounded by a constant. 
We next discuss what happens if we relax either of these two restrictions.

\emph{Tree-structured polymatrix games with unrestricted number of actions:} A bimatrix game is basically a tree-structured polymatrix game with two players. \citet{ChenSettlingJACM} showed that there exists no FPTAS for bimatrix games with an unrestricted number of actions unless all problems in PPAD are polynomial-time solvable. In this paper, we bound the number of actions by a constant. We should also note the main motivation behind graphical games,
as originally introduced by~\citet{Kearns_et_al_2001}: compact/succinct
representations 
where the representation sizes do not depend exponentially in $n$, but are instead
exponential in $k$ and linear in $n$. As~\citet{Kearns_et_al_2001}
stated, if $k \ll n$, we obtain exponential gains in representation
size. Thus, it is $n$ and $k$ the parameters of main interest in standard graphical games; the
parameter $m$ is of secondary interest. Indeed, even
\citet{Kearns_et_al_2001} concentrate on the case of 
$m=2$.

\emph{Graphical (not necessarily tree-structured) polymatrix games with bounded number of actions:} 
\citet{rubinstein15} showed that for $\epsilon = 10^{-8}$ and $m = 10^4$, computing an $\epsilon$-MSNE for an $n$-player game is PPAD-hard. This hardness proof involves the construction of graphical (non-tree) polymatrix games. Therefore, the result carries over to $n$-player graphical polymatrix games. This lower bound result shows that graph structures that are more complex than trees are intractable (under standard assumptions) even for constant $m$ and small but constant $\epsilon$.

\section{Preliminaries, Background, and Notation}





Denote by $a \equiv (a_1,a_2,\ldots,a_n)$ an $n$-dimensional vector
and by $a_{-i} \equiv (a_1,\ldots,a_{i-1},a_{i+1},\ldots,a_n)$ the
same vector without the $i$-th component. Similarly, for every set $S \subset
[n] \equiv \{1,\ldots,n\}$, denote by $a_S \equiv (a_i : i \in S)$ the
(sub-)vector formed from $a$ using 
exactly the 
components of $S$. 
$S^c \equiv [n] - S$ denotes the complement of $S$, and $a \equiv
(a_S,a_{S^c}) \equiv (a_i,a_{-i})$ for every $i$. If $A_1,\ldots,A_n$
are sets, denote by $A \equiv \times_{i \in [n]} A_i$, $A_{-i} \equiv
\times_{j \in [n] - \{i\}} A_j$ and $A_S \equiv \times_{j \in S}
A_j$. 
To simplify the presentation, whenever we have a difference of a set $S$
with a singleton set $\{i\}$, we often abuse notation and
denote by $S - i \equiv S - \{i\}$.

\subsection{GMhG Representation}


\begin{definition}
A {\em graphical multi-hypermatrix game (GMhG)\/} is defined by a set
$V$ of $n$ players and the followings for each player $i \in V$:
\begin{itemize}

\item a set of {\em
  actions\/} or pure strategies $A_i$;
  
\item  a set $\mathcal{C}_i \subset
2^V$ of {\em local cliques\/} or \emph{local
  hyperedges} 
 such that if $C \in \mathcal{C}_i$ then $i \in C$, and two additional sets defined based on $\mathcal{C}_i$:
 
 \begin{itemize}
 \item $i$'s neighborhood $N_i \equiv \cup_{C \in \mathcal{C}_i} C$ (the set of players, including $i$, that affect $i$'s payoff) and
 \item $\Neigh_i \equiv \{ j \in V \mid i \in N_j, j \neq i\}$ (the set of players, not including $i$, affected by $i$);
 \end{itemize}
 
\item a set $\{M'_{i,C} : A_C \to \bbbr  \mid C \in \mathcal{C}_i \}$ of {\em local-clique payoff matrices\/}; and

\item the {\em local\/} and {\em global payoff matrices\/} $M'_i : A_{N_i} \to \bbbr$ and $M_i : A \to \bbbr$ of $i$ 
defined as $M'_i(a_{N_i}) \equiv \sum_{C \in \mathcal{C}_i} M'_{i,C}(a_{C})$ and $M_i(a) \equiv M'_i(a_{N_i})$, respectively.

\end{itemize}
  
\end{definition}

We denote by $\kappa_i \equiv |\mathcal{C}_i|$
  and $\kappa \equiv \max_i \kappa_i$ the number of hyperedges of player $i$
  and the maximum number of hyperedges over all players,
  respectively. Similarly, we denote $\kappa'_i \equiv \max_{C \in
    \mathcal{C}_i} |C|$
  and $\kappa' \equiv \max_i \kappa'_i$ the size of the biggest hyperedge of player $i$
  and the size of the biggest hyperedge over all players,
  respectively. Also, for consistency with the graphical games literature,
  we denote by $k_i \equiv |N_i|$ and $k \equiv \max_i
  k_i$ the size of the neighborhood of the primal graph induced by the
  local hyperedges of $i$ and the maximum neighborhood
  size over all players, respectively. 

\begin{figure}
\centering
  \includegraphics[width=0.3\textwidth]{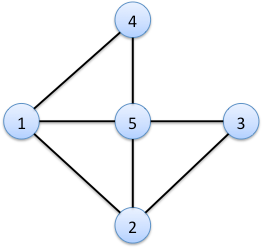}
  \caption{\textbf{The primal graph of an example GMhG.} The endpoints of each edge belong to some common hyperedge. The sets of local hyperedges of players 1 to 5 are: $\mathcal{C}_1 = \{ \{1, 2\}\}$, $\mathcal{C}_2 = \{ \{2, 5\}\}$, $\mathcal{C}_3 = \{ \{3, 5\}\}$, $\mathcal{C}_4 = \{ \{1, 4\}, \{4, 5\}\}$, and $\mathcal{C}_5 = \{ \{1, 5\}, \{5, 2, 3\}\}$. 
  The neighborhood of 5 is $N_5 = \{1, 2, 3, 5\}$ and the set of players affected by 5 is $\Neigh_5 = \{2, 3, 4\}$. The local payoff of 5 is the sum of payoffs from the hyperedges $\{1, 5\}$ and $\{5, 2, 3\}$. For player 5, $\kappa_5 = 2$ and $\kappa'_5 = 3$. For the whole graph, $\kappa = 2$ and $\kappa' = 3$.
  }
  \label{fig:ex1}
\end{figure}%

Fig.~\ref{fig:ex1} illustrates some of the above terminology. The GMhG shown there (without the actual payoff matrices) is not a graphical game, because in a graphical game each $\mathcal{C}_i$ must be singleton (i.e., only one local hyperedge for each node $i$, which corresponds to $N_i$). This GMhG is not a polymatrix game either, because not all local hyperedges consist of only 2 nodes. Furthermore, the GMhG is not a hypergraphical game~\cite{Papadimitriou:2008:CCE:1379759.1379762}, because the local hyperedges are not symmetric (player 1's local hyperedge has 2 in it, but 2's local hyperedge does not have 1).

The
representation sizes of GMhGs, polymatrix games, and graphical games are $O(n \, \kappa
  \, m^{\kappa'})$, $O(n \, k \,
  m^2 )$, and  $O(n \,
  m^k )$, respectively.
%
%
%

\paragraph{Normalizing the Payoff Scale.} 
The dominant mode of approximation in game theory is additive approximation~\cite{Lipton_et_al_2003,Daskalakis_et_al_2007,deligkas2014,barman15}. 
For $\epsilon$ to be truly meaningful as a \emph{global} additive approximation
parameter, the payoffs of \emph{all} players must be brought to the
\emph{same} scale. The convention in the literature 
(see, e.g., \citetp{deligkas2014}) is to assume that (1) each player's local payoffs 
are spread between 0 and 1, with the local payoff being exactly 0 for
some joint action and exactly 1 for another; and (2) the local-\emph{clique}
payoffs (i.e., entries in the payoff matrices) are between 0 and 1. 
Here, we relax the second assumption;
that is, we can handle matrix entries that are
negative or larger than 1. 
Indeed, because of the \emph{additive} nature of the local payoffs
in GMhGs, the ``$[0,1]$
assumption'' on those payoffs may require that some of the
local-clique payoffs contain values $<0$ or $>1$. This is a
\emph{key} aspect of payoff scaling, and in turn the approximation problem, that
often does not get proper attention.
We have a much milder assumption that the \emph{maximum spread} of 
local-clique payoffs (or matrix entries) of each player is bounded by a constant. 
We allow this constant to be different for different players.

Note that the equilibrium conditions are invariant to affine
transformations. In the case of graphical games with local payoff
matrices represented in tabular/matrix/normal-form, it is convention
to assume 
that the maximum and minimum local
payoff values of each player are $0$ and $1$, respectively. 
This assumption is without loss of generality, because for any general
graphical game, 
we can find the minimum and maximum local payoff of each player \emph{efficiently} 
and thereby make these $0$ and $1$, respectively through affine transformations.

While doing this 
for GMhGs in general is intractable in the worst
case, it is computationally efficient for GMhGs whose local
hypergraphs have bounded hypertree-widths. For instance, 
the payoffs of a graphical polymatrix game can be normalized in polynomial time 
to achieve the first assumption above. 
To do that, we define the following terms. 
\begin{align*}
u_i &= \max_{a_i} \sum_{j \in
  \mathcal{N}_i} \max_{a_j} M'_{i,j}(a_i,a_j)\\
l_i &= \min_{a_i} \sum_{j \in
  \mathcal{N}_i} \min_{a_j} M'_{i,j}(a_i,a_j)    
\end{align*}
It is evident from the
last expression that we can efficiently compute each of those values for each
$i$ via
dynamic programming (DP) in time $|A_i| (|\mathcal{N}_i| + |A_i|) = O(mk +
m^2)$, and compute all the values for all $i$ in time $2 \, \sum_i
|A_i| (|\mathcal{N}_i| + |A_i|) = O(m |E| + m^2)$.

Despite such exceptions, in general, we do not have much of a choice but to
assume that the payoffs of all players are in the same scale, so that
using a global 
$\epsilon$ is meaningful.
For any local-clique payoff hypermatrix $M'_{i,C}$, we define the following notation on the maximum payoff, minimum payoff, and the largest spread of payoffs in that hypermatrix, respectively.
\begin{align*}
& u_{i,C} \equiv \max_{a_C \in A_C} M'_{i,C}(a_C)\\
& l_{i,C} \equiv \min_{a_C \in A_C} M'_{i,C}(a_C)\\
& R_{i,C} \equiv u_{i,C} - l_{i,C}
\end{align*}

\paragraph{Example.} The following example shows that restricting the values of
the local-clique hypermatrices to $[0,d_i]$
while keeping the maximum and minimum values of the local payoff
functions of each player $i$ to be $1$ and $0$, respectively, loses
generality (e.g., some local-clique payoffs may be negative).The reason is for some games there is no affine transformation that would satisfy both of these
conditions while maintaining exactly the same equilibrium conditions.
Let $b,c > 0$, and $0 < \gamma < \frac{1}{3}$.
\begin{align*}
M'_{1,2} =& \left[ 
\begin{array}{cc}
1+2b & 1+2b-\gamma\\
-2c+\gamma & -2c
\end{array}
\right]\\
M'_{1,3} =& M'_{1,4} = \left[ 
\begin{array}{cc}
-b & -b-\gamma\\
c+\gamma & c
\end{array}
\right]
\end{align*}
\begin{align*}
u_{1,2} =& 1+2b, u_{1,3} = u_{1,4} = c+\gamma\\
l_{1,2} =& -2c, l_{1,3} = l_{1,4} = -b\\
R_{1,2} =& 1+2b+2c, R_{1,3} = R_{1,4} = b+c+\gamma
\end{align*}
\begin{tabular}{cccc|c}
\end{tabular}
\begin{align*}
M'_1(a_1,a_2,a_3,a_4) = 
\begin{cases}
1, & \text{if $a_1=a_2=a_3=a_4=1$,}\\
1-\gamma, & \text{if $a_1=1$ and exactly one of $a_2$, $a_3$, or $a_4$
  $=2$,}\\
1-2\gamma, & \text{if $a_1=1$ and exactly two of $a_2$, $a_3$, or $a_4$
  $=2$,}\\
1-3\gamma, & \text{if $a_1=1$ and $a_2=a_3=a_4=2$,}\\
3\gamma, & \text{if $a_1=2$ and $a_2=a_3=a_4=1$,}\\
2\gamma, & \text{if $a_1=2$ and exactly two of $a_2$, $a_3$, or $a_4$
  $=1$,}\\
\gamma, & \text{if $a_1=2$ and exactly one of $a_2$, $a_3$, or $a_4$
  $=1$,}\\
0, & \text{if $a_1=a_2=a_3=a_4=2$,}
\end{cases}
\end{align*}

\section{Discretization Scheme: Simple Version}

In contrast with earlier discretization schemes~\cite{Kearns_et_al_2001}, we allow different discretization sizes for different players. Also, in contrast with recent schemes~\cite{ortiz14}, we discretize \emph{both} the probability space (Definition~\ref{def:prob})  and the payoff space (Definition~\ref{def:payoff}).

\begin{definition}
\label{def:prob}
{\em (Individually-uniform mixed-strategy discretization
  scheme)\/} Let $I = [0,1]$ be the uncountable set of the possible values
of the probability $p_i(a_i)$ of each action $a_i$ of each
player $i$. Discretize $I$ by a finite grid defined by the set
$\Itilde_i = \{0, \tau_i, 2 \tau_i, \ldots, (s_i-1) \tau_i, 1\}$
with interval
$\tau_i = 1/s_i$ for some
integer $s_i > 0$. Thus the {\em mixed-strategy-discretization size\/} is
$|\Itilde_i| = s_i + 1$. We only consider mixed
strategies $q_i$ such that $q_i(a_i) \in \Itilde_i$ for all
$a_i$, and $\sum_{a_i} q_i(a_i) = 1$. The {\em induced
  mixed-strategy discretized space\/} of joint mixed strategies is
$\Itilde \equiv \times_{i \in V} \Itilde_i^{\, |A_i|}$, subject to the
individual normalization constraints. 
\end{definition}


\begin{definition}
\label{def:payoff}
{\em (Individually-uniform expected-payoff discretization
  scheme)\/} Let $C'_i \equiv C'_i(\C'_i) \equiv \cup_{C \in \C'_i}
  C$. Define the following two terms.
\begin{align*}
\bar{l}_i \equiv & \min_{\C'_i \subset \C_i, \C'_i \neq \emptyset} \min_{a_{C'_i}} \sum_{C \in
  \C'_i} M'_{i,C}(a_C) \geq \min_{\C'_i \subset \C_i, \C'_i \neq \emptyset} \sum_{C \in
  \C'_i} l_{i,C} \geq \min_{\C'_i \subset \C_i, \C'_i \neq \emptyset} |\C'_i| l_i \geq
                   \min\left( \kappa_i l_i , l_i \right) \; \\
&\text{(The last inequality above considers the cases of negative and non-negative $l_i$.)}\\
\bar{u}_i \equiv & \max_{\C'_i \subset \C_i, \C'_i \neq \emptyset} \max_{a_{C'_i}} \sum_{C
                   \in \C'_i} M'_{i,C}(a_C) \leq \max_{\C'_i \subset \C_i, \C'_i \neq \emptyset} \sum_{C \in
  \C'_i} u_{i,C} \leq \max_{\C'_i \subset \C_i, \C'_i \neq \emptyset}
                   |\C'_i| u_i \leq \kappa_i u_i \; 
\end{align*}
Let $I_i= \left[\bar{l}_i ,\bar{u}_i \right]$ denote an
interval containing every possible expected
payoff values that each player $i$ can receive from each local-clique
payoff matrix $M'_{i,C}(p_C)$, where 
$p_C \in \Itilde_C$ (i.e., $p_C$ is in the grid). 
Discretize $I_i$ by a finite grid defined by the set $\Itilde'_i =
\{\bar{l}_i, \bar{l}_i + \tau'_i, \bar{l}_i +  2 \tau'_i, \ldots,
\bar{l}_i+  (s'_i-1) \tau'_i, \bar{u}_i \}$ 
with interval $\tau'_i = \bar{R}_i/s'_i$ for some integer
$s'_i > 0$, where $\bar{R}_i \equiv \bar{u}_i - \bar{l}_i \leq
\max(\kappa_i R_i, \kappa_i u_i - l_i)$.
Thus the {\em expected-payoff-discretization size\/} is
$|\Itilde'_i| = s'_i + 1$. Then, for any $B \subset C \in
\mathcal{C}_i$, we would only consider
an expected-payoff $\Mtilde'_{i,C}(a_B,q_{C-B})$ in the discretized grid that is 
closest to the exact local-clique expected payoff
$M'_{i,C}(a_B,q_{C-B})$. 
More formally, $\Mtilde'_{i,C}(a_B,q_{C-B}) = \argmin_{r \in
  \Itilde'_i} |r
-  M'_{i,C}(a_B,q_{C-B})| \equiv \Proj{M'_{i,C}(a_B,q_{C-B})}$.
The {\em induced expected-payoff discretized space\/} over all 
local-cliques of all players is $\Itilde' \equiv \times_{i
  \in V}
\left( \Itilde'_i \right)^{\kappa_i}$. 
\end{definition}

\citet{hau_aaai15} use a similar idea in the setting
  of \emph{interdependent defense (IDD) games}, where each of
  $n$ sites
  has a binary pure-strategy set, and a specific instance of the
  general setting in which the attacker has $n+1$ pure
  strategies. 
  The reason why the attacker has
  $n+1$ pure strategies is because, in the particular instance of IDD games
  that~\citet{hau_aaai15} consider, the attacker can attack \emph{at most} one site at a
  time, simultaneously. 
In contrast, the potential multiplicity of actions of \emph{all}
players poses one of the main challenges in our case, particularly
because of the non-tabular/non-normal-form representation of the
general GMhGs, which is exponential in the size of the largest hyper-edge
over all players neighborhood hyper-graphs.

\section{A GMhG-Induced CSP: Simple Version}
\label{sec:CSP}

Consider the following CSP induced by a GMhG: 
\begin{itemize}
\item \emph{Variables:} for all $i$ and $a_i$, a variable $p_{i,a_i}$
  corresponding to the mixed-strategy/probability that player $i$
  plays pure strategy $a_i$ and, for all $C \in \C_i$, a
  variable $S_{i,C,a_i}$ corresponding to some \emph{partial sum} of the
  expected payoff of player $i$ based on an ordering of the
  local hyperedge elements
  of $\C_i$. Formally, if $\mathcal{P}_i \equiv
  \bigcup_{a_i} \{ p_{i,a_i} \}$ and $\mathcal{S}_{i,C}
  \equiv \bigcup_{a_i} \{ S_{i,C,a_i} \}$, then the set of all variables is
  $\bigcup_{i} \left( \mathcal{P}_i \bigcup_{C \in \mathcal{C}_i}  \mathcal{S}_{i,C} \right)$.
\item \emph{Domains:} the domain of each variable $p_i(a_i)$ is
  $\widetilde{I}_i$, while that of each partial-sum variable $S_{i,C,a_i}$ is $\Itilde'_i$.
\item \emph{Constraints:} for each $i$: 
\begin{enumerate} 
\item \emph{Best-response and partial-sum expected local-clique payoff:}
  We first compute a hyper-tree decomposition of the local hypergraph induced
  by
  hyperedges $\mathcal{C}_i$. We then \emph{order} the set of local-cliques $\mathcal{C}_i$
of each player $i$ such that $\mathcal{C}_i \equiv
\{C_i^1,C_i^2,\ldots,C_i^{\kappa_i}\}$.  
The superscript
denotes the corresponding order of the local-cliques of player $i$. We make sure that 
 the order is consistent with the hypertree decomposition of the
local hypergraph, in the standard (non-serial) DP-sense
used in constraint and probabilistic graphical models~\citep{dechter_book,Koller+Friedman:09}. For any $a_i$:
\begin{enumerate}
\item 
\(
\sum_{a'_i} p_{i,a'_i} S_{i,C_i^{\kappa_i},a'_i} \geq
S_{i,C_i^{\kappa_i},a_i} - \frac{2}{3} \epsilon;
\)
\item 
  $S_{i,C_i^1,a_i} =
  \Mtilde_{i,C_i^1}(a_i,p_{C_i^{1}-i})$, and
  for $l = 2,\ldots,\kappa_i$,
\[
S_{i,C_i^l,a_i} =
  \Mtilde_{i,C_i^l}(a_i,p_{C_i^l - i}) + S_{i,C_i^{l-1},a_i}.
\]
\end{enumerate}
We call (a) the best-response constraint and (b) the partial-sum expected local-clique payoff constraint.
\item \emph{Normalization:} $\sum_{a_i} p_{i,a_i} = 1$.
\end{enumerate}
\end{itemize}

The number of variables of the CSP is $O\left( n \, m \, \kappa
\right)$. The size of each domain $\widetilde{I}_i$ is $O\left( s
\right)$, where $s \equiv \max_i s_i$. The size of each domain
$\widetilde{I}_i'$ is 
$O\left( s' \right)$, 
where $s' \equiv \max_i s_i'$. 
The computation of each $\Mtilde_{i,C_i^l}(a_i,p_{C_i^l - i})$ in 1(b) above, which takes time
$O( s^{\kappa'-1})$, 
dominates the running time to build the constraint set. 
The total number of constraints is 
$O\left( n \, m  \, \kappa \right)$.
The maximum number of variables in any constraint
is 
$O(m \, \kappa')$. 
Given a hyper-tree decomposition, the amount of time to build the
constraint set using a tabular representation is $O(n \, m \, \kappa \,
s^{m \, \kappa'} (s')^m )$, which is the \emph{representation size of the
  GMhG-induced CSP}. 

\subsection{Correctness of the GMhG-Induced CSP}


We use the following Lemma of \citet{ortiz14}. Note that our results do not follow directly from this Lemma, since we also discretize the payoff space. Furthermore, for tree-structured polymatrix games, \citet{ortiz14}'s running time depends on $(\frac{k}{\epsilon})^k$ when $m$ is bounded by a constant, whereas ours is polynomial in the maximum neighborhood size $k$.

\begin{lemma}{\bf (Sparse MSNE Representation Theorem)}
\label{lem:sd}
For any GMhG and any $\epsilon$ such that \[0  < \epsilon \leq 2 \,
\min_{i \in V} \, \frac{\sum_{C \in \mathcal{C}_i} R_{i,C} \, (|C| -
  1)}{\max_{C' \in \mathcal{C}_i} |C'| - 1} \, ,\] a (uniform)
discretization with 
\begin{align*} 
s_i = \left\lceil \frac{2 \, |A_i| \, \max_{j
      \in \Neigh_i} \sum_{C \in \mathcal{C}_j} R_{j,C} \, (|C| - 1) }{
    \epsilon }\right\rceil 
\end{align*}
for each player $i$ is sufficient to guarantee that for {\em every\/} MSNE of the game, its closest (in $\ell_\infty$ distance) joint mixed strategy in the induced discretized space 
is also an $\epsilon$-MSNE.
\end{lemma}

We next present our sparse-representation theorem, 
where we discretize the partial sums of expected local-clique payoffs.
\begin{theorem}{ {\bf (Sparse Joint MSNE and Expected-Payoff
      Representation Theorem)}}
\label{thm:joint_sd}
Consider any GMhG and any $\epsilon$, \[0  < \epsilon \leq 2 \,
\min_{i \in V} \, \frac{\sum_{C \in \mathcal{C}_i} R_{i,C} \, (|C| -
  1)}{\max_{C' \in \mathcal{C}_i} |C'| - 1} \, .\] 
Setting, for all
players $i$, the pair $(\tau_i,\tau_i')$ defining the joint
(individually-uniform) mixed-strategy and expected-payoff
discretization of player $i$ such that
\[
\tau_i = \frac{\epsilon}{6 \, |A_i| \, \max_{j \in \Neigh_i} \sum_{C \in \mathcal{C}_j} R_{j,C} \, (|C| - 1) }
\]
and
\[
\tau'_i = \frac{\epsilon}{3 \, \kappa_i} \; ,
\]
so that the discretization sizes 
\begin{align*}
s_i = \left\lceil \frac{6 \, |A_i| \, \max_{j \in \Neigh_i}
    \sum_{C \in \mathcal{C}_j} R_{j,C} \, (|C| - 1) }{\epsilon}
\right\rceil 
\end{align*}
and
\[
s'_i = \left\lceil \frac{3 \, \bar{R}_i \, \kappa_i}{\epsilon} \right\rceil 
\]
for each mixed-strategy probability
and expected payoff value, respectively,
is sufficient to guarantee that for {\em every\/} MSNE of the game, its closest (in $\ell_\infty$ distance) joint mixed strategy in the induced discretized space 
is a solution of the GMhG-induced CSP, and 
that any solution to the GMhG-induced CSP (in discretized probability and payoff space) is an 
$\epsilon$-MSNE of the game.
\end{theorem}
\begin{proof}
For the first part of the theorem, 
let $p'$ be an MSNE of the GMhG. Let $p$ be the mixed strategy closest, in
$\ell_\infty$, 
to $p'$ in the grid induced by the combination of the discretizations
that each $\tau_i$ generates. For all $i$ and $a_i$,
set $p^*_{i,a_i} = p_i(a_i)$; and for all $i$ and $a_i$, first set $S^*_{i,C_i^1,a_i} =
\Mtilde_{i,C_i^1}(a_i,p^*_{C_i^1-i})$, and then recursively for $l=2,\ldots,\kappa_i$,
set $S^*_{i,C_i^l,a_i} =
\Mtilde_{i,C_i^l}(a_i,p^*_{C_i^l-i}) +
S^*_{i,C_i^{l-1},a_i}$. The resulting assignment satisfies the normalization
constraint of the CSP, by the definition of a mixed strategy. The
assignment also satisfies the partial-sum expected local-clique
payoffs by construction. Thus, we are left to prove that the best-response constraint is satisfied.
By the setting of $\tau_i$ and Lemma~\ref{lem:sd}, we have that
$p$ is an $(\epsilon/3)$-MSNE, and thus also an
$\epsilon$-MSNE. In addition, for all $i$ and $a_i$, we have the
following sequence of inequalities:
\begin{align*}
\sum_{a'_i} p_i(a'_i) \sum_{C \in \mathcal{C}_i}
M'_{i,C}(a'_i,p_{C-i}) \geq 
\sum_{C \in \mathcal{C}_i}
M'_{i,C}(a_i,p_{C-i}) - \frac{\epsilon}{3}.
\end{align*}
\begin{align}
\label{ineq:part1}
\sum_{a'_i} p^*_{i,a'_i} \sum_{l=1}^{\kappa_i} 
M'_{i,C_i^l}(a'_i,p^*_{C_i^l-i}) \geq 
\sum_{l=1}^{\kappa_i} 
M'_{i,C_i^l}(a_i,p^*_{C_i^l-i}) - \frac{\epsilon}{3}.
\end{align}
By the definition of $\Mtilde_{i,C}$, for all $i$ and $C \in
\mathcal{C}_i$, we have that for all $a_i$ and $l=1,\ldots,\kappa_i$,
\begin{align*}
\Mtilde_{i,C_i^l}(a'_i,p^*_{C_i^l-i}) - \frac{\tau'_i}{2} \leq 
M'_{i,C_i^l}(a'_i,p^*_{C_i^l-i})
\leq
\Mtilde_{i,C_i^l}(a'_i,p^*_{C_i^l-i}) + \frac{\tau'_i}{2} \; .
\end{align*}
Applying the last inequality to (\ref{ineq:part1})
and by unraveling the construction of the CSP
assignment, we have
\begin{align*}
\sum_{a'_i} p^*_{i,a'_i} \sum_{l=1}^{\kappa_i} 
\left( \Mtilde_{i,C_i^l}(a'_i,p^*_{C_i^l-i}) + \frac{\tau'_i}{2} \right)
\geq
\sum_{l=1}^{\kappa_i} \left(
\Mtilde_{i,C_i^l}(a_i,p^*_{C_i^l-i}) - \frac{\tau'_i}{2} \right) - \frac{\epsilon}{3}
\end{align*}
and
\[
\sum_{a'_i} p^*_{i,a'_i} S^*_{i,C_i^{\kappa_i},a'_i} \geq 
S^*_{i,C_i^{\kappa_i},a_i} - \kappa_i \, \tau'_i - \frac{\epsilon}{3} \; .
\]
Rearranging the terms, and plugging in 
$\kappa_i \, \tau'_i = \frac{1}{3}\epsilon$ we get 

\[
\textstyle
\sum_{a'_i} p^*_{i,a'_i} S^*_{i,C_i^{\kappa_i},a'_i} \geq 
S^*_{i,C_i^{\kappa_i},a_i} - \frac{2}{3}\epsilon  \; .
\]

Hence, the assignment $(p^*,S^*)$ also satisfies the best-response constraints (1(a) of CSP) and is a
solution to the GMhG-induced CSP.

Now, for the second part of the theorem, suppose $(p^*,S^*)$ is a solution of the GMhG-induced CSP. Then, by the combination of the best-response and partial-sum
expected local-clique payoff constraints, we
have that, for all $i$ and $a_i$,
\[
\sum_{a'_i} p^*_{i,a'_i} S^*_{i,C_i^{\kappa_i},a'_i} \geq
S^*_{i,C_i^{\kappa_i},a_i} - \frac{2}{3} \epsilon \; ,
\]
\[
S^*_{i,C_i^1,a_i} =
  \Mtilde_{i,C_i^1}(a_i,p^*_{C_i^{|\mathcal{C}_i|}-i}) \, ,
\]
\[
S^*_{i,C_i^l,a_i} =
  \Mtilde_{i,C_i^l}(a_i,p^*_{C_i^l - i}) + S^*_{i,C_i^{l-1} - i,a_i} \; .
\]
This in turn implies that for all $i$ and $a_i$, we can obtain the
following sequence of inequalities:
\[
\sum_{a'_i} p^*_{i,a'_i} \sum_{l=1}^{\kappa_i}
\Mtilde_{i,C_i^l}(a'_i,p^*_{C_i^l - i}) \geq
\sum_{l=1}^{\kappa_i}
\Mtilde_{i,C_i^l}(a_i,p^*_{C_i^l - i}) - \frac{2}{3}\epsilon
\]
\[
\sum_{a'_i} p^*_{i,a'_i} \sum_{C \in \mathcal{C}_i}
\Mtilde_{i,C}(a'_i,p^*_{C - i}) \geq
\sum_{C \in \mathcal{C}_i}
\Mtilde_{i,C}(a_i,p^*_{C - i}) - \frac{2}{3} \epsilon
\]

\begin{align*}
\sum_{a'_i} p^*_{i,a'_i} \sum_{C \in \mathcal{C}_i}
\left( M'_{i,C}(a'_i,p^*_{C - i}) + \frac{\tau'_i}{2} \right) \geq
\sum_{C \in \mathcal{C}_i}
\left( M'_{i,C}(a_i,p^*_{C - i}) - \frac{\tau'_i}{2} \right) -
  \frac{2}{3} \epsilon
\end{align*}
\begin{align*}
\sum_{a'_i} p^*_{i,a'_i} \sum_{C \in \mathcal{C}_i}
M'_{i,C}(a'_i,p^*_{C - i}) \geq 
\sum_{C \in \mathcal{C}_i}
M'_{i,C}(a_i,p^*_{C - i}) - \kappa_i \, \tau'_i - \frac{2}{3} \epsilon
\end{align*}
\begin{align*}
\sum_{a'_i} p^*_{i,a'_i} \sum_{C \in \mathcal{C}_i}
M'_{i,C}(a'_i,p^*_{C - i}) \geq
\sum_{C \in \mathcal{C}_i}
M'_{i,C}(a_i,p^*_{C - i}) - \frac{1}{3}\epsilon - \frac{2}{3} \epsilon
\end{align*}
\begin{align*}
\textstyle
\sum_{a'_i} p^*_{i,a'_i} \sum_{C \in \mathcal{C}_i}
M'_{i,C}(a'_i,p^*_{C - i}) \geq 
\sum_{C \in \mathcal{C}_i}
M'_{i,C}(a_i,p^*_{C - i}) - \epsilon
\end{align*}
Hence, the corresponding joint mixed-strategy 
$p^*$ is an $\epsilon$-MSNE of the GMhG.
\end{proof}
\begin{claim}
\label{cla:joint_sd}
Within the context of Theorem~\ref{thm:joint_sd}, we have
\begin{align*} 
s_i = & O\left( \frac{m \, \kappa'  \, \max_{j \in \Neigh_i} \sum_{C \in
  \C_j} R_{j,C}}{\epsilon} \right) = O\left( \frac{m \, \kappa' \,  R
        }{\epsilon} \right) \; , \\
s'_i = & O\left( \frac{\bar{R}_i \, \kappa_i}{\epsilon} \right) =
         O\left( \frac{ R' \, \kappa }{\epsilon} \right) \; ,
\end{align*}
where $R \equiv \max_{j \in V} \sum_{C \in \C_j} R_{j,C}$ and
$R' \equiv \max_{j \in V} \bar{R}_j$.
If all the
ranges $R_{j,C}$'s are
bounded by a constant, then 
\begin{align*} 
s_i =& O\left( \frac{m \, \kappa' \, \kappa }{\epsilon} \right) \;
       \text{ and}\\
s'_i =& O\left( \frac{ \kappa^2 }{\epsilon} \right) \; .
\end{align*}
\end{claim}
\begin{proof}
First, when all the
ranges $R_{j,C}$'s are
bounded by a constant, 
we have $R = O(\kappa)$.
Furthermore, $\bar{R}_j = \bar{u}_j - \bar{l}_j \le \kappa_j u_j - \min( \kappa_j l_j , l_j )$. When $l_j < 0$, $\bar{R}_j \le \kappa_j (u_j - l_j)$ and hence $\bar{R}_j = O(\kappa_j)$. For the other case of $l_j \ge 0$, $\bar{R}_j \le \kappa_j u_j - l_j \le \kappa_j u_j$. Since $u_j - l_j$ is bounded by a constant and $l_j \ge 0$, $u_j$ must also be bounded by a constant and hence $\bar{R}_j = O(\kappa_j)$. Therefore, $R' = O(\kappa)$. Since both $R$ and $R'$ are $O(\kappa)$, we obtain the bounds on $s_i$ and $s'_i$.
\end{proof}
Note that if $R'$ is bounded by a constant, then $l_i = \omega(-1/\kappa_i)$.

\section{CSP-Based Computational Results}
\label{sec:results}

The CSP formulation in the previous section leads us to the following
computational results based on well-known algorithms for solving
CSPs~\cite[Ch. 5]{Russell_and_Norvig_2003}, and the application of equally well-known
computational results for them~\citep{dechter_book,gottlob14,gottlob16}.

\begin{theorem}
\label{the:hw}
There exists an algorithm that, given as input a number $\epsilon > 0$
and an $n$-player GMhG with maximum local-hyperedge-set size $\kappa$
and maximum number of actions $m$, and whose corresponding CSP has a
hypergraph with hypertree-width $w$, computes an $\epsilon$-MSNE of
the GMhG in time $[ n \left( m \, \kappa \, \kappa' / \epsilon
\right)^{m \kappa'}]^{O(w)}$.
\end{theorem}

 
For GMhGs with bounded hypertree width $w$, the following corollary 
establishes our main CSP-based result.
\begin{corollary}
There exists an algorithm that, given as input a
GMhG with bounded $w$, outputs an $\epsilon$-MSNE in polynomial time
in the size of the input and $1/\epsilon$, for any
$\epsilon > 0$; hence, the algorithm is an FPTAS. If, instead, we have
$w =
O(\mbox{polylog}(n))$, then the algorithm is a quasi-PTAS.
\end{corollary}

Theorem~\ref{the:hw} also implies that we can compute an $\epsilon$-MSNE of a tree-structured polymatrix game in 
$O(n \left( m k / \epsilon \right)^{2m})$. \emph{Note that the running time is polynomial in the maximum neighborhood size $k$.}

The following results are in term of the primal-graph representation
of the GMhG-induced CSP.

\begin{theorem}
There exists an algorithm that, given as input a number $\epsilon > 0$
and an $n$-player GMhG with maximum number of actions $m$,
primal-graph treewidth $w'$ of the corresponding CSP, maximum local-hyperedge-set size $\kappa$, and
maximum local-hyperedge size $\kappa'$, computes an $\epsilon$-MSNE of the game
in time $2^{O(w')} n \log(n) + n [ \left( m \, \kappa \, \kappa' / \epsilon
\right)^m ]^{O(w')}$. 
\end{theorem}


\begin{corollary}
There exists an FPTAS for computing an approximate MSNE in $n$-player
GMhGs with corresponding $m$, $\kappa$, and $\kappa'$ all bounded by constants, independent
of $n$, and primal-graph treewidth $w' = O(\log(n))$.
\end{corollary}
\begin{corollary}
There exists an algorithm that, given as input an $n$-player
polymatrix GG with a tree graph, maximum neighborhood size $k$, and maximum
number of actions $m$, computes an $\epsilon$-MSNE of the polymatrix
GG in time $2^{O(m)} n \log(n) + n \left( m k / \epsilon
\right)^{O(m)} $. If $m$ is
bounded by a constant, then the algorithm is an FPTAS. If, instead, $m=O(\polylog(n))$, then the algorithm is a quasi-PTAS. 
\end{corollary}


\section{DP for $\epsilon$-MSNE Computation}
\label{sec:DP}

We present a DP algorithm in the
context of the special, but still important class of tree-structured \emph{polymatrix} games. This is for
simplicity and clarity, and as we later discuss, is
without loss of generality. 
We first designate an arbitrary node as the
\emph{root} of the tree and define the notion of
\emph{parents} and \emph{children} nodes as follows. 
For any node/player $i$,
we
denote by $\pa(i)$ the \emph{single parent of any non-root node} in
the tree and
 by $\Ch(i)$ the \emph{children} of node $i$ in the
root-designated-induced directed tree. 
If $i$ is the root, then $\pa(i)$ is undefined. If $i$ is a leaf, then $\Ch(i) = \emptyset$. 


The two-pass algorithm is similar in spirit to {\em TreeNash} \citep{Kearns_et_al_2001}, 
except that (1) here the messages are $\{-\infty,0\}$, instead of bits $\{0,1\}$; and (2) more distinctly, 
our algorithm implicitly passes messages about the partial-sum of expected payoffs across the siblings.
%

\textbf{Collection Pass.} For each non-root node $i$, 
we denote by $j =
\pa(i)$. We order $\Ch(i)$ as 
$o_1,\ldots,o_{|\Ch(i)|}$. We then apply the
following DP bottom-up (i.e., from leaves to root).
We give an intuition before giving the formal specification. The message $T_{i \to j}(p_i,p_j)$ is 0 iff it is ``OK'' for $i$ to play $p_i$ when $i$'s parent $j$ plays $p_j$ (the notion of OK recursively makes sure that $i$'s children are also OK). The message $B_i(p_i,p_j,S_{o_{|\Ch(i)|}})$ is 0 iff $i$'s best response to $j$ playing $p_j$ is $p_i$, given that $i$ gets a combined payoff of $S_{o_{|\Ch(i)|}}$ from its children. The message $R_{o_{l}}(p_i,S_{o_{l}})$ can be thought of as being implicitly passed from $i$'s child $o_l$ to the next (and back to $i$ from the last child $o_{|\Ch(i)|}$). $R_{o_{l}}(p_i,S_{o_{l}})$ is 0 iff $S_{o_{l}}$ is the maximum payoff that $i$ can get from its first $l$ children when $i$ plays $p_i$ and those children are OK with that. Fig.~\ref{fig:alg} illustrates the message passing.

\begin{figure}
\centering
  \includegraphics[width=0.4\textwidth]{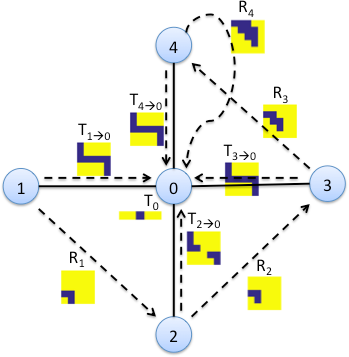}
  \caption{\textbf{DP on a 5-node star polymatrix game.} Solid lines represent edges, broken lines show the final round of message passing. The endpoints of every edge is playing a matching pennies game between them. 
  The visualization of $T_{1 \to 0}$, for example, plots $p_1(a_1 = 0)$ on $x$-axis and $p_0(a_0 = 0)$ on $y$-axis. Dark grid points denote OK (i.e., $T_{1 \to 0} = 0$) and light grid points not OK. The $R_i$ tables are 3-dimensional. Here we only show one slice of $R_i$ values corresponding to $p_0 = (0.5, 0.5)$. The $x$-axis represents $S_i(a_0 = 0)$ (partial sum up to the $i$-th child when player 0 plays 0) and $y$-axis $S_i(a_0 = 1)$. A $0.1$-MSNE computed for this instance is: $p_0 = p_3 = (0.5,0.5)$, $p_1=(0.75,0.25)$, $p_2 = p_4 = (0, 1)$.
  }
  \label{fig:alg}
\end{figure}%

Formally, 
for each arc $(j,i)$ in the
designated-root-induced directed tree (i.e., $j$ is the parent of $i$), and $(p_i,p_j)$
a mixed-strategy pair in the induced grid:
\begin{align*}
T_{i \to j}(p_i,p_j) = & 
\max_{S_{o_{|\Ch(i)|}}} B_i(p_i,p_j,S_{o_{|\Ch(i)|}})
  + R_{o_{|\Ch(i)|}}(p_i,S_{o_{|\Ch(i)|}}) \\
W_{i \to j}(p_i,p_j) = & 
\textstyle 
\argmax_{S_{o_{|\Ch(i)|}}} B_i(p_i,p_j,S_{o_{|\Ch(i)|}})
  + R_{o_{|\Ch(i)|}}(p_i,S_{o_{|\Ch(i)|}})
\end{align*}
where $B_i(p_i,p_j,S_{o_{|\Ch(i)|}}) =$
\begin{align*}
\sum_{a_i} \log\left(\mathds{1} \left[ \sum_{a'_i} p_i(a'_i) \left( \Mtilde_{i,j}(a'_i,p_j) + S_{o_{|\Ch(i)|}}(a'_i)
     \right) \geq \Mtilde_{i,j}(a_i,p_j) + S_{o_{|\Ch(i)|}}(a_i) -
      \epsilon \right] \right)
\end{align*}
and, for $l = 1,\ldots,|\Ch(i)|$,
\begin{align*}
V_{o_l}(S_{o_l},p_{o_l},S_{o_{l-1}}) =& \sum_{a_i}
                                        \log\left(\mathds{1}\left[ S_{o_l}(a_i)
                                        = \Mtilde_{i,o_l}(a_i,p_{o_l})
                                        + S_{o_{l-1}}(a_i) \right]
  \right)\\
F_{o_l}(p_i,S_{o_l},p_{o_l},S_{o_{l-1}}) =& T_{o_l \to i}(p_{o_l},p_i)
                                            +
                                            R_{o_{l-1}}(p_i,S_{o_{l-1}}) +
                                            V_{o_l}(S_{o_l},p_{o_l},S_{o_{l-1}})\\
R_{o_l}(p_i,S_{o_l}) =& \max_{p_{o_l},S_{o_{l-1}}}
                            F_{o_l}(p_i,S_{o_l},p_{o_l},S_{o_{l-1}})\\
W_{o_l}(p_i,S_{o_l}) =& \argmax_{p_{o_l},S_{o_{l-1}}}
                            F_{o_l}(p_i,S_{o_l},p_{o_l},S_{o_{l-1}})
                        \; .
\end{align*}

Following are the boundary conditions: $R_{o_0} \equiv 0$ and $S_{o_0} \equiv 0$,
so that
$F_{o_1}(p_i,S_{o_1},p_{o_1},S_{o_0}) \equiv
  F_{o_1}(p_i,S_{o_1},p_{o_1}) = T_{o_1 \to i}(p_{o_1},p_i)$.
  If $i$ is the root, then 
  $T_{i \to j}(p_i,p_j) \equiv T_i(p_i)$ and
  $W_{i \to j}(p_i,p_j) \equiv W_i(p_i)$. 
If $i$ is a leaf, $T_{i \to j} (p_i,p_j)$ takes a simpler, non-recursive form.

\textbf{Assignment Pass.} For root $i$, set $p^*_i \in
\argmax_{p_i} T_i(p_i)$ and $S^*_{o_{|\Ch(i)|}} \in W_i(p^*_i)$.
Then recursively apply the following
assignment process starting at $o_{|\Ch(i)|}$: for
$l=|\Ch(i)|,\ldots,1$, set $(p^*_{o_l},S^*_{o_{l-1}}) \in
W_{o_l}(p^*_i,S^*_{o_l})$.

\subsection{The Running Time of the DP Algorithm}

A running-time analysis of the DP algorithm presented above yields the
following theorem, which is one of our main algorithmic results of
this paper.
\begin{theorem}
The DP algorithm computes an $\epsilon$-MSNE in a
graphical polymatrix game with a tree graph in time $n\left( \frac{m \, k}{\epsilon}
\right)^{O(m)}$.
\end{theorem}
\begin{corollary}
The DP algorithm is an FPTAS to compute an $\epsilon$-MSNE in an $n$-player
graphical polymatrix game with a tree graph and a bounded number of actions $m$.
If 
$m=O(\polylog(n))$, 
then the DP algorithm is a quasi-PTAS.
\end{corollary}

\section{Further Refinement}
We describe a more refined alternative to the
GMhG-induced CSP that reduces the dependency on
$\kappa$. We are also able to restrict the expected payoff grid to [0, 1], even though the local-clique payoffs of a player may fall outside of [0, 1].
The main idea is to evaluate the expressions
involving the expected
local-clique payoffs matrices $\Mtilde_{i,C}(a_i,p_{C-i})$ in a
smart way by decomposing the sum involving the expectation,
considering one player mixed-strategy at a time, and
projecting to the discretized payoff space after evaluating each term
in the sum. This approach gives us an
FPTAS for tree graphical games (in normal form) and
bounded number of actions, for which the best known approximation
result to-date is a quasi-PTAS. 
We present the main result below.

\begin{theorem}
\label{sec:GGThm}
There exists a DP algorithm that computes an $\epsilon$-MSNE in a
tree graphical game in time 
$n \; \left( \frac{mk}{\epsilon}
\right)^{3m+2} \; O\left( \left( m^k
  \right)^2\right)$.
If $m$ is bounded, then the
running time is $\mathrm{poly}\left( n\,
  m^k,\frac{1}{\epsilon} \right)$ and the algorithm is an FPTAS. If $m=O(\polylog(n))$, 
then this algorithm is a quasi-PTAS.
\end{theorem}

\subsection{Discretization Scheme: Refined Version}

\begin{definition}
\label{def:payoff_refined}
Consider the uncountable set $I = [0,1]$. For each player $i$, 
approximate $I$ by a finite grid defined by the set $\Itilde'_i =
\{0, \tau'_i, 2 \tau'_i, \ldots, (s'_i-1) \tau'_i, 1\}$ of values
separated by the same distance $\tau'_i = 1/s'_i$ for some integer
$s'_i$. Thus $|\Itilde'_i| = s'_i + 1$. Then, for any value $v \in I$,
we now define its projection $\widetilde{v} \equiv \Proji{v}$ to
$\Itilde'_i$ such that $|\widetilde{v} -  v| \leq \tau'_i/2$. 
\end{definition}

\subsection{GMhG-Induced CSP: Refined Version}
\label{sec:CSPGG}

We now present a more complex generalization of the simpler GMhG-induced CSP.
\begin{itemize}
\item \emph{Variables:} for all $i$, $a_i$, a variable $p_{i,a_i}$
  corresponding to the mixed-strategy/probability that player $i$
  plays pure strategy $a_i$ and, for all $C \in \C_i$, a
  variable $S_{i,C,a_i}$ corresponding to some \emph{scale-normalized partial sum} of the
  expected payoff of player $i$ based on an ordering of the
  local-clique/hyperedge elements
  of $\mathcal{C}_i$, and given an ordering of $C - i =
  \left\{o_1,o_2,\ldots,o_{|C|-1}\right\}$, a variable
  $E_{i,C,a_{C-[o_t]}}$ corresponding to some \emph{scale-normalized partial
    conditional}
  expected payoff of player $i$; that is, formally, if $\mathcal{P}_i \equiv
  \bigcup_{a_i} \{ p_{i,a_i} \}$, $\mathcal{S}_{i,C}
  \equiv \bigcup_{a_i} \{ S_{i,C,a_i} \}$, and $\mathcal{E}_{i,C}
  \equiv \bigcup_{t=1}^{|C|-1} \bigcup_{a_{C-[o_t]}} \{ E_{i,C,a_{C-[o_t]}} \}$, then the set of all variables is
  $\bigcup_i \left( \mathcal{P}_i \cup \bigcup_{C \in \mathcal{C}_i}
    \left( \mathcal{S}_{i,C} \bigcup \mathcal{E}_{i,C} \right)\right)$.
\item \emph{Domains:} the domain of each variable $p_i(a_i)$ is
  $\widetilde{I}_i$, while that of each partial-sum variable
  $S_{i,C,a_i}$ and each partial-conditional-expectation variable $E_{i,C,a_{C-[o_t]}}$ is $\Itilde'_i$.
\item \emph{Constraints:} for each $i$, 
\begin{enumerate} 
\item \emph{Best-response and partial-sum expected local-clique
    payoff:}
Compute a hyper-tree decomposition of the local hypergraph induced
  by
  hyperedges $\mathcal{C}_i$; then \emph{order} the set of local-cliques $\mathcal{C}_i$
of each player $i$ such that $\mathcal{C}_i \equiv
\{C_i^1,C_i^2,\ldots,C_i^{\kappa_i}\}$, where the superscript
denotes the corresponding order of the local-cliques of player $i$,
and the order is consistent with the hypertree decomposition of the
local hypergraph, in the standard (non-serial) DP-sense
used in constraint and probabilistic graphical models~\citep{Dechter_2003,Koller+Friedman:09}; and 
 for any $a_i$,
\begin{enumerate}
\item 
\[
\sum_{a'_i} p_{i,a'_i} S_{i,C_i^{\kappa_i},a'_i} \geq
S_{i,C_i^{\kappa_i},a_i} - \frac{2}{3 \, \sum_{C \in \C_i} R_{i,C}}  \epsilon
\]
\item $S_{i,C_i^1,a_i} = E_{i,C_i^1,a_i}$, and for $l = 2,\ldots,\kappa_i$, 
\[
  S_{i,C_i^l,a_i} = \Proji{\frac{R_{i,C_i^l} \, E_{i,C_i^l,a_i} + \left(\sum_{r=1}^{l-1} R_{i,C_i^r} \right) \, S_{i,C_i^{l-1},a_i}}{\sum_{r=1}^l R_{i,C_i^r}}}
\]
\item for each set $C \in \mathcal{C}_i$, order the
  elements of $C - i$ such that $C - i =
  \left\{o_1,o_2,\ldots,o_{|C|-1}\right\}$, then set
  $E_{i,C,a_{C-[o_1]}} = \Proji{\sum_{a_{o_1}} p_{o_1,a_{o_1}}
    \left( \frac{M'_{i,C}(a_{o_1}, a_{C-[o_1]})-l_{i,C}}{R_{i,C}} \right) }$ and for $t =
  2,\ldots,|C| -1$,
\[
E_{i,C,a_{C-[o_t]}} = \Proji{\sum_{a_{o_t}} p_{o_t,a_{o_t}} E_{i,C,a_{C-[o_{t-1}]}}}
\]
\end{enumerate}
\item \emph{Normalization:} $\sum_{a_i} p_{i,a_i} = 1$
\end{enumerate}
\end{itemize}

The number of variables of the CSP is $O\left( n \, \kappa \, m^{\kappa'} 
\right)$, which is larger than the version of the CSP presented in the
main body, but \emph{exactly} the worst-case representation size of the
GMhG. 
The normalization of $M'_{i,C}(a_C)$ in 1(c) above takes
constant time, assuming we have precomputed $l_{i,C}$ and $u_{i,C}$,
which take $O(m^{|C|})$, which is the representation size of $M'_{i,C}$, for each $i$ and $C \in \C_i$, of course. 
The total number of constraints is 
 $O\left( n \, \kappa \, m^{\kappa'} 
\right)$, which is also larger than the version of the CSP presented in the
main body, but \emph{exactly} the worst-case representation size of the
GMhG. 
The maximum number of variables in any constraint
is $O(m)$, which is smaller than the version of the CSP presented in the
main body by a factor of $\kappa'$. 
Given a hyper-tree decomposition, the amount of time to build the
constraint set using a tabular representation is $O(n \, \kappa \,
m^{\kappa'+ 1} \, 
s^m (s')^{m+1} )$. 

In summary, the \emph{representation size of the
  GMhG-induced CSP} presented above,
  using a tabular representation, is 
$O(n \, \kappa \,
m^{\kappa'+ 1} \, 
s^m (s')^{m+1} )$. 

\emph{Note the key reduction in the dependence on $\kappa'$ from the
  analogous expression given for the
CSP: the parameter $\kappa'$ only appears in the exponent
of $m$, as it also does in the representation size of the GMhG, and
not in the exponent of $s$.}

\begin{theorem}{ {\bf (Sparse Joint MSNE and Expected-Payoff
      Representation Theorem: Refined Version)}}
\label{thm:joint_sd_refined}
Consider any GMhG and any $\epsilon$, \[0  < \epsilon \leq 2 \,
\min_{i \in V} \, \frac{\sum_{C \in \mathcal{C}_i} R_{i,C} \, (|C| -
  1)}{\max_{C' \in \mathcal{C}_i} |C'| - 1} \, .\] 
Setting, for all
players $i$, the pair $(\tau_i,\tau_i')$ such that
\[
\tau_i = \frac{\epsilon}{6 \, |A_i| \, \max_{j \in \Neigh_i}
    \sum_{C \in \mathcal{C}_j} R_{j,C} \, (|C| - 1)}
\]
and
\[
\tau'_i = \frac{\epsilon}{3 \, \left( 
    \sum_{C \in \mathcal{C}_i} R_{i,C} \, (|C| + \kappa_i - l) \right)} \; ,
\]
so that the discretization sizes 
\begin{align*}
s_i = \left\lceil \frac{1}{\tau_i} \right\rceil 
\end{align*}
and
\[
s'_i = \left\lceil \frac{1}{\tau'_i} \right\rceil 
\]
respectively,
is sufficient to guarantee that for {\em every\/} MSNE of the game, its closest (in $\ell_\infty$ distance) joint mixed strategy in the induced discretized space 
is a solution of the refined version of the GMhG-induced CSP, and 
that any solution to the refined version of the GMhG-induced CSP (in discretized space) is an 
$\epsilon$-MSNE of the game.
\end{theorem}
\begin{proof}
Let $p'$ be an MSNE of the GMhG. Let $p$ be the mixed strategy closest, in
$\ell_\infty$, 
to $p'$ in the grid induced by the combination of the discretizations
that each $\tau_i$ generates. For all $i$ and $a_i$,
set $p^*_{i,a_i} = p_i(a_i)$. For all $i$ and $C \in \C_i$, such that
$C-i \equiv \{o_1,o_2,\ldots,o_{|C|-1} \}$, first set 
\[
E^*_{i,C,a_{C-[o_1]}} \equiv \Proji{ \frac{M'_{i,C}(a_{C-[o_1]},p_{i,o_1}) - l_{i,C}}{R_{i,C}}}.
\]
Then for $t=2,3,\ldots,|C|-1$, recursively set
\[
E^*_{i,C,a_{C-[o_t]}} \equiv \Proji{ \sum_{a_{o_t}} p_{i,a_{o_t}}
  E^*_{i,C,a_{C-[o_{t-1}]}}} \; .
\]
Finally, given $\C_i \equiv \left\{
  C_i^1,C_i^2,\ldots,C_i^{\kappa_i} \right\}$, for all $i$ and $a_i$, first set 
\[
S^*_{i,C_i^1,a_i} \equiv E^*_{i,C_i^1,a_i}.
\]
Then for all $l=2,3,\ldots,\kappa_i$, recursively set
\[
S^*_{i,C_i^l,a_i} \equiv \Proji{\frac{R_{i,C_i^l} \, E^*_{i,C_i^l,a_i} +
    \left(\sum_{r=1}^{l-1} R_{i,C_i^r} \right) \,
    S^*_{i,C_i^{l-1},a_i}}{\sum_{r=1}^l R_{i,C_i^r}}} \; .
\]
The resulting assignment satisfies constraints 1(b), 1(c), and 2 of the CSP by construction. We next show that constraint 1(a) is also satisfied. By the setting of $\tau_i$ and Lemma~\ref{lem:sd}, we have that
$p$ is an $(\epsilon/3)$-MSNE, and thus also an
$\epsilon$-MSNE. In addition, for all $i$ and $a_i$, we have the
following sequence of inequalities:
\begin{align*}
\sum_{a'_i} p_i(a'_i) \sum_{C \in \mathcal{C}_i}
M'_{i,C}(a'_i,p_{C-i}) \geq 
\sum_{C \in \mathcal{C}_i}
M'_{i,C}(a_i,p_{C-i}) - \frac{\epsilon}{3}.
\end{align*}
\begin{align*}
\sum_{a'_i} p^*_{i,a'_i} \sum_{l=1}^{\kappa_i} 
M'_{i,C_i^l}(a'_i,p^*_{C_i^l-i}) \geq 
\sum_{l=1}^{\kappa_i} 
M'_{i,C_i^l}(a_i,p^*_{C_i^l-i}) - \frac{\epsilon}{3}.
\end{align*}
For all $i$ and $C \in \C_i$, by the Principle of Mathematical
Induction we can show that for all
$t=1,2,\ldots,|C|-1,$ we have 
\begin{align*}
\left| E^*_{i,C,a_{C-[o_t]}} - \frac{M'_{i,C}(a_{C-[o_t]},p^*_{[o_t]})
  - l_{i,C}}{R_{i,C}} \right| \leq t \, \frac{\tau'_i}{2} \; .
\end{align*}
Similarly, we can show that for all $i$, $a_i$, and $t=1,2,\ldots,\kappa_i$, 
\begin{align*}
\left| S^*_{i,C_i^t,a_i} - \frac{\sum_{l=1}^t
  M'_{i,C_i^l}(a_i,p^*_{C_i^l-i}) - l_{i,C_i^l}}{\sum_{l=1}^t
  R_{i,C_i^l}} \right| \leq  \frac{\sum_{l=1}^t R_{i,C_i^l}
  (|C_i^l| + t - l) }{\sum_{l=1}^t
  R_{i,C_i^l}} \cdot \frac{\tau'_i}{2} \; .
\end{align*}
From the last condition, we obtain
\begin{align*}
M_i'(p_{N(i)}) \leq \left( \sum_{l=1}^{\kappa_i}
  R_{i,C_i^l} \right) \left( \sum_{a'_i} p^*_{i,a'_i} S^*_{i,C_i^{\kappa_i},a'_i} \right)  +
  \sum_{l=1}^{\kappa_i} l_{i,C_i^l} + \sum_{l=1}^{\kappa_i} R_{i,C_i^l}
  (|C_i^l| + \kappa_i - l) \frac{\tau'_i}{2} 
\end{align*}
and
\begin{align*}
M_i'(a_i,p_{\Neigh(i)}) \geq \left( \sum_{l=1}^{\kappa_i}
  R_{i,C_i^l} \right) S^*_{i,C_i^{\kappa_i},a_i} +
  \sum_{l=1}^{\kappa_i} l_{i,C_i^l} - \sum_{l=1}^{\kappa_i} R_{i,C_i^l}
  (|C_i^l| + \kappa_i - l) \frac{\tau'_i}{2} \; .
\end{align*}
Combining, we obtain the following sequence of inequalities.
\begin{align*}
\left( \sum_{l=1}^{\kappa_i}
  R_{i,C_i^l} \right) \left( \sum_{a'_i} p^*_{i,a'_i} S^*_{i,C_i^{\kappa_i},a'_i} \right)  -
  \sum_{l=1}^{\kappa_i} l_{i,C_i^l} + \sum_{l=1}^{\kappa_i} R_{i,C_i^l}
  (|C_i^l| + \kappa_i - l) \frac{\tau'_i}{2}  \geq \\
\left( \sum_{l=1}^{\kappa_i}
  R_{i,C_i^l} \right) S^*_{i,C_i^{\kappa_i},a_i} -
  \sum_{l=1}^{\kappa_i} l_{i,C_i^l} - \sum_{l=1}^{\kappa_i} R_{i,C_i^l}
  (|C_i^l| + \kappa_i - l) \frac{\tau'_i}{2} - \frac{1}{3} \epsilon\\
\sum_{a'_i} p^*_{i,a'_i} S^*_{i,C_i^{\kappa_i},a'_i} \geq S^*_{i,C_i^{\kappa_i},a_i} - \frac{\sum_{l=1}^{\kappa_i} R_{i,C_i^l}
  (|C_i^l| + \kappa_i - l)}{\sum_{l=1}^{\kappa_i}
  R_{i,C_i^l}} \tau'_i - \frac{1}{3 \, \sum_{l=1}^{\kappa_i}
  R_{i,C_i^l}} \epsilon
\end{align*}
Substituting for $\tau'_i$, we obtain the following equation.
\begin{align*}
\tau'_i \sum_{l=1}^{\kappa_i} R_{i,C_i^l}
  (|C_i^l| + \kappa_i - l) = \frac{1}{3} \epsilon \;
\end{align*}
We can rewrite the above equation as follows.
\begin{align*}
\tau'_i \sum_{C \in \mathcal{C}_i} R_{i,C} (|C| + \kappa_i - l) = \frac{1}{3} \epsilon \;
\end{align*}
This completes the proof that $(p^*,S^*,E^*)$ is a solution to the refined
GMhG-induced CSP.

Now, for the second part of the theorem, suppose $(p^*,S^*,E^*)$ is a
solution of the refined GMhG-induced CSP. Then, the following holds for all $i$ and $a_i$.
\begin{align*}
\sum_{a'_i} p^*_{i,a'_i} S^*_{i,C_i^{\kappa_i},a'_i} \geq
S^*_{i,C_i^{\kappa_i},a_i} - \frac{2}{3 \, \sum_{l=1}^{\kappa_i}
  R_{i,C_i^l} } \epsilon \\
\sum_{a'_i} p^*_{i,a'_i} \left(  \frac{\sum_{l=1}^{\kappa_i}
  M'_{i,C_i^l}(a'_i,p^*_{C_i^l-i}) - l_{i,C_i^l}}{\sum_{l=1}^{\kappa_i}
  R_{i,C_i^l}} + \frac{\sum_{l=1}^{\kappa_i} R_{i,C_i^l}
  (|C_i^l|-1 + \kappa_i - l) \frac{\tau'_i}{2}}{\sum_{l=1}^{\kappa_i}
  R_{i,C_i^l}}\right) \geq \\
\left(  \frac{\sum_{l=1}^{\kappa_i}
  M'_{i,C_i^l}(a_i,p^*_{C_i^l-i}) - l_{i,C_i^l}}{\sum_{l=1}^{\kappa_i}
  R_{i,C_i^l}} - \frac{\sum_{l=1}^{\kappa_i} R_{i,C_i^l}
  (|C_i^l|-1 + \kappa_i - l) \frac{\tau'_i}{2}}{\sum_{l=1}^{\kappa_i}
  R_{i,C_i^l}} \right) - \frac{2}{3 \, \sum_{l=1}^{\kappa_i}
  R_{i,C_i^l} } \epsilon \\
M'_i(p^*_{N(i)}) \geq M'_i(a_i,p^*_{\Neigh{i}}) - \left( \sum_{l=1}^{\kappa_i} R_{i,C_i^l}
  (|C_i^l|-1 + \kappa_i - l) \tau'_i \right) - \frac{2}{3} \epsilon\\
M'_i(p^*_{N(i)}) \geq M'_i(a_i,p^*_{\Neigh{i}}) - \frac{1}{3} \epsilon - \frac{2}{3} \epsilon\\
M'_i(p^*_{N(i)}) \geq M'_i(a_i,p^*_{\Neigh{i}}) - \epsilon
\end{align*}
Hence, the corresponding joint mixed-strategy 
$p^*$ is an $\epsilon$-MSNE of the GMhG.
\end{proof}
\begin{claim}
\label{cla:joint_sd_refined}
Within the context of Theorem~\ref{thm:joint_sd_refined}, we have
\begin{align*} 
s_i = & O\left( \frac{m \, \kappa'  \, \max_{j \in \Neigh_i} \sum_{C \in
  \C_j} R_{j,C}}{\epsilon} \right) = O\left( \frac{m \, \kappa'  \, \max_{j \in \Neigh_i} R'_j}{\epsilon} \right) = O\left( \frac{m \, \kappa' \,  R'
        }{\epsilon} \right) \; \text{ and} \\
s'_i = & O\left( \frac{R'_i \, (\kappa'_i + \kappa_i)}{\epsilon}
         \right) 
\; ,
\end{align*}
where $R'_i \equiv \sum_{C \in
  \C_i} R_{i,C}$ and $R' \equiv \max_{j \in V} R'_j$.  
If all the
ranges $R_{j,C}$'s are
bounded by a constant, then we have $R' = O(\kappa)$ and $R'_i =
O(\kappa_i) = O(\kappa)$,
which implies
\begin{align*} 
s_i = O\left( \frac{m \, \kappa' \, \kappa }{\epsilon} \right) \;
       \text{ and  }
s'_i = O\left( \frac{ \kappa \, (\kappa' + \kappa) }{\epsilon} \right) \; .
\end{align*}
\end{claim}

\subsection{Sparse-Discretization-Based DP for Graphical Games in Normal-Form with Tree Graphs}
\label{sec:DPGG}

This appendix is analogous to the DP presented in the main body,
but deals with \emph{normal-form GGs,} instead of polymatrix GGs. We refer
the reader to the introduction to the DP framework for general context
and notation. Because we are dealing with standard graphical games in
normal form, we can assume without loss of generality that the local
payoff matrices are properly normalized to $[0,1]$.

\paragraph{Collection Pass.} Recursively, for each node $i$ in the
induced directed tree, relative to the root, denote by $j =
\pa(i)$. Order $\Ch(i)$ and denote the resulting node order by
$o_1,\ldots,o_{|\Ch(i)|}$. Apply the
following DP from leaves to root: for each arc $(j,i)$ in the
designated-root-induced directed tree, and $(p_i,p_j)$
a mixed-strategy pair in the induced grid,
\begin{align*}
T_{i \to j}(p_i,p_j) = & \max_{S_{o_{|\Ch(i)|}}} B_i(p_i,p_j,E_{o_{|\Ch(i)|}})
  + R_{o_{|\Ch(i)|}}(p_i,E_{o_{|\Ch(i)|}}) \text{ and}\\
W_{i \to j}(p_i,p_j) = & \argmax_{S_{o_{|\Ch(i)|}}}
                         B_i(p_i,p_j,E_{o_{|\Ch(i)|}}) +
                         R_{o_{|\Ch(i)|}}(p_i,E_{o_{|\Ch(i)|}}) \; ,
\end{align*}
where $B_i(p_i,p_j,E_{o_{|\Ch(i)|}}) =$
\begin{align*}
\sum_{a_i}
\log \mathbbm{1}\left[ \sum_{a'_i,a'_j} p_i(a'_i) p_j(a'_j)
    E_{o_{|\Ch(i)|}}(a'_i,a'_j) \geq \sum_{a'_j} p_j(a'_j)
    E_{o_{|\Ch(i)|}}(a_i,a'_j) -
      \epsilon \right]
\end{align*}
and, for $l = 1,\ldots,|\Ch(i)|$, 
\begin{align*}
V_{o_l}(E_{o_l},p_{o_l},E_{o_{l-1}}) =& \sum_{a_{N_i-[o_l]}}
                                        \log \mathbbm{1}\left[E_{o_l}(a_{N_i-[o_l]})
                                        = \Proji{\sum_{a_{o_l}}
                                        p_{o_l}(a_{o_l}) 
    E_{o_l}(a_{N_i-[o_{l-1}]})} \right] \; , \\
F_{o_l}(p_i,E_{o_l},p_{o_l},E_{o_{l-1}}) =&T_{o_l \to i}(p_{o_l},p_i)
                                            +
                                            R_{o_{l-1}}(p_i,E_{o_{l-1}})
                                            +
                                            V_{o_l}(E_{o_l},p_{o_l},E_{o_{l-1}})
  \; , \\
R_{o_l}(p_i,E_{o_l}) =& \max_{p_{o_l},E_{o_{l-1}}}
                            F_{o_l}(p_i,E_{o_l},p_{o_l},E_{o_{l-1}})
                        \; , \text{ and} \\
W_{o_l}(p_i,E_{o_l}) =& \argmax_{p_{o_l},E_{o_{l-1}}}
                            F_{o_l}(p_i,E_{o_l},p_{o_l},E_{o_{l-1}})
                        \; .
\end{align*}
Note that we are using the following boundary conditions for
simplicity of presentation: $R_{o_0} \equiv 0$ and, for all $a_{N_i}$, $E_{o_0}(a_{N_i}) \equiv \Proji{M'_i(a_{N_i})}$,
so that
\begin{align*}
F_{o_1}(p_i,E_{o_1},p_{o_1},E_{o_0}) \equiv&
  F_{o_1}(p_i,E_{o_1},p_{o_1})\\
 =&
 T_{o_1 \to i}(p_{o_1},p_i) +\\
&
  \sum_{a_{N_i-o_1}} \log\mathbbm{1}\left[E_{o_1}(a_{N_i-o_1}) = 
      \Proji{\sum_{a_{o_1}} p_{o_1}(a_{o_1}) E_{o_0}(a_{N_i})}\right] \; .
\end{align*}
  If $i$ is the designated root, then, because there is no
  corresponding parent $j$, we have
  $T_{i \to j}(p_i,p_j) \equiv T_i(p_i)$ and
  $W_{i \to j}(p_i,p_j) \equiv W_i(p_i)$.

\paragraph{Assignment Pass.} For the root $i$, set $p^*_i \in
\argmax_{p_i} T_i(p_i)$ and $E^*_{o_{|\Ch(i)|}} \in W_i(p^*_i)$, where
$o_{|\Ch(i)|}$ is the last node in the order of the root's children $\Ch(i)$. Then recursively apply the following
assignment process starting at $o_{|\Ch(i)|}$: for
$l=|\Ch(i)|,\ldots,1$, set $(p^*_{o_l},E^*_{o_{l-1}}) \in
W_{o_l}(p^*_i,E^*_{o_l})$.

\emph{Note that for the case of polymatrix GGs the DP above would be
  essentially the same as that presented in
  in the main body.}

\subsubsection{Running Time of DP Algorithm for Tree Graphical Games in Normal-Form}

The worst-case running-time for message passing at each node $i$ is
\begin{align*}
& \left( \frac{m k}{\epsilon} \right)^{2m} \times \sum_{l=1}^{|\Ch(i)|}
  O\left( \left( \frac{m k}{\epsilon} \right)^{m} \left( m^{|\Ch(i)|-l}
  \right)^{2} m^2 \left( \frac{k}{\epsilon} \right)^{2} \right) \\
= & \left( \frac{m k}{\epsilon} \right)^{3m+2} \times \sum_{r=0}^{|\Ch(i)|-1}
  O\left( \left(m^2\right)^r \right) \\
= & \left( \frac{m k}{\epsilon} \right)^{3m+2} 
  O\left( \frac{\left(m^2\right)^{|\Ch(i)|} -1 }{m^2 - 1} \right) \\
= & \left( \frac{m k}{\epsilon} \right)^{3m+2} 
  O\left( \left(m^2\right)^{(k_i-2)-1}  \right) \\
= & \left( \frac{m k}{\epsilon} \right)^{3m+2} 
  O\left( m^{2 k_i-6} \right)
\end{align*}
if $i$ is an internal node, $O\left( m^2 \left( \frac{m k}{\epsilon}
  \right)^{2m} \right)$
if $i$ is a leaf, and $\left( \frac{m k}{\epsilon} \right)^{2m+2} 
  O\left( m^{2 k_i-4} \right)$ if $i$ is the root,
from which the running-time result of Theorem 5 follows.

\section{Concluding Remarks}
We have presented tractable algorithms for computing $\epsilon$-MSNE in tree-structured GMhGs when the number of actions is bounded. 
The implications of our results can best be highlighted by considering a very simple 101-node star polymatrix game with a constant number of actions. For computing an $\epsilon$-MSNE of this game, the algorithm of \citet{Kearns_et_al_2001} takes $O(((\frac{2^{k+2}k \log k}{\epsilon})^2)^k)$ time (here $k = 100$), the algorithm of \citet{ortiz14} takes $O((\frac{k}{\epsilon})^k)$ time, and ours takes $O(poly(\frac{k}{\epsilon}))$ time and thereby solves a 15-year-old open problem. We conclude by emphasizing that our DP algorithm is simple to implement and that simplicity is a strength of this work.

\bibliographystyle{theapa}
\bibliography{games}

\end{document}